\documentclass{cccg12}
\usepackage{graphicx,amssymb}

\newcommand{\seg}{\overline}

\title{Hidden Mobile Guards in Simple Polygons\footnote{An abstract version of this paper was presented at the 21st Fall Workshop on Computational Geometry, 2011. Research supported in part by NSF grants CCF-0830734 and CBET-0941538.}}

\author{Sarah Cannon\thanks{Department of Computer Science,
	Tufts University, Medford, MA, USA {\tt scanno01@cs.tufts.edu}}
	\and
        Diane L. Souvaine\thanks{Department of Computer Science,
        Tufts University, Medford, MA, USA {\tt dls@cs.tufts.edu}}
	\and
        Andrew Winslow\thanks{Department of Computer Science,
        Tufts University, Medford, MA, USA {\tt awinslow@cs.tufts.edu}}}

\index{Cannon, Sarah}
\index{Souvaine, Diane}
\index{Winslow, Andrew}

\begin{document}
\thispagestyle{empty}
\maketitle

\begin{abstract}
We consider guarding classes of simple polygons using mobile guards (polygon edges and diagonals) under the constraint that no two guards may see each other.
In contrast to most other art gallery problems, existence is the primary question: does a specific type of polygon admit \emph{some} guard set?
Types include simple polygons and the subclasses of orthogonal, monotone, and starshaped polygons.
Additionally, guards may either exclude or include the endpoints (so-called \emph{open} and \emph{closed} guards).
We provide a nearly complete set of answers to existence questions of open and closed edge, diagonal, and mobile guards in simple, orthogonal, monotone, and starshaped polygons, with some surprising results. 
For instance, every monotone or starshaped polygon can be guarded using hidden open mobile (edge or diagonal) guards, but not necessarily with hidden open edge or hidden open diagonal guards. 
\end{abstract}

\section{Definitions}

We define the \emph{boundary} of a polygon $P$ (denoted $\partial P$) as a simple polygonal chain consisting of a sequence of vertices specified in counterclockwise order, and the open set enclosed by $\partial P$ to be the \emph{interior} of $P$ (denoted ${\rm int}(P)$).
An \emph{edge} $e = \seg{pq}$ of the polygon is an interval of $\partial P$ between consecutive vertices $p, q$, and a \emph{diagonal} $d = \seg{rs}$ of $P$ is a straight line segment between non-consecutive vertices $r, s$ of $\partial P$ such $d - \{r, s\} \in {\rm int}(P)$, i.e. the portion of $d$ excluding its endpoints lies in the interior of $P$.

We consider guarding ${\rm int}(P)$ using a subset of the edges and diagonals of $P$.
A guard $g$ \emph{sees} or \emph{guards} a location $l$ in the polygon if $l$ is \emph{weakly visible}~\cite{Avis-1981} from the guard: there exists a point $p \in g$ such that the interior of the segment $lp$ lies in the interior of the polygon.
Edges and diagonals selected as guards are called \emph{edge guards} and \emph{diagonal guards}, respectively, and a \emph{mobile guard}~\cite{ORourke-1983} is either an edge or a diagonal guard.
If a set $S$ of edges and diagonals of $P$ is such that every location in the interior of $P$ is seen by at least one guard in $S$, then $S$ is a \emph{guard set} of $P$ and $P$ is said to \emph{admit} a guard set.
A \emph{closed guard set} includes the vertices at both ends of each edge or diagonal.
If all endpoints are excluded, the guard set is called an \emph{open guard set}.
 
In addition to \emph{simple polygons} or simply \emph{polygons}, we consider a number of special classes of polygons.
An \emph{orthogonal polygon} is a polygon that can be rotated such that all edges are parallel to the x- or y-axis.
A \emph{monotone polygon} is a polygon that can be rotated such that the portion of the polygon intersecting any vertical line consists of a connected interval.
A \emph{starshaped polygon} is a polygon that can be translated such that an interior point coincides with the origin and sees all locations in the interior of the polygon, and the \emph{kernel} of the polygon is the set of all points in the polygon with this property.
These three classes (along with convex and spiral polygons) are described by O'Rourke~\cite{ORourke-1987} in the context of guarding problems as being ``usefully distinguished in the literature.''

Finally, we add the constraint that a guard set is \emph{hidden}: no pair of guards in the set see each other.
Here a pair of guards $g_1, g_2$ in a polygon $P$ can see each other if there exists a pair of points $p \in g_1, q \in g_2$ such that $pq - \{p,q\} \in {\rm int}(P)$.

\section{Introduction}

Edge, diagonal, and mobile guards in polygons have been studied extensively in the past.
Avis and Toussaint~\cite{Avis-1981} considered the case where a single closed edge is sufficient to guard the entire polygon.
Shortly after, Toussaint gave an example of a polygon whose smallest closed edge guard set is $\lfloor n/4 \rfloor$~\cite{ORourke-1983} and conjectured that an edge guard set of this size is sufficient for any polygon.
O'Rourke~\cite{ORourke-1983} showed that closed mobile guard sets of size $\lfloor n/4 \rfloor$ are sometimes necessary and always sufficient for polygons.
For closed diagonal guards, Shermer~\cite{Shermer-1992} has shown that guard sets of size $\lfloor (2n + 2)/7 \rfloor$ are necessary for some polygons, and no polygon requires a guard set of size greater than $\lfloor (n-1)/3 \rfloor$.

More recently, open edge guards were suggested by Viglietta~\cite{Viglietta-2011} and studied by Benbernou et al.~\cite{Benbernou-2011} and T\'{o}th et al.~\cite{Toth-2011}, who showed that open edge guard sets of size $\lfloor n/3 \rfloor$ and $\lfloor n/2 \rfloor$ are sometimes necessary and always sufficient for simple polygons. 

The study of hidden guards began with Shermer~\cite{Shermer-1989} who gave several results, including examples of polygons that are not guardable using hidden vertex guards.
The study of hidden edges has only been initiated recently by Kranakis et al.~\cite{Kranakis-2009} who showed that computing the largest hidden open edge set in a polygon (ignoring guarding) cannot be approximated within an arbitrarily small constant factor unless P = NP. 
In the same theme, Kosowoski et al.~\cite{Kosowski-2006} have studied cooperative mobile guards, where each guard is \emph{required} to be seen by another guard.
Such a constraint is the opposite of hiddenness, which \emph{forbids} any guard from seeing any other guard. 

Here we evaluate the existence of hidden edge, diagonal, and mobile guard sets for simple polygon classes.
A summary of results is seen in Table~\ref{tab:results}.

\begin{table}[ht]
\centering
\begin{tabular}{|c|c|c|c|c|c|}
\hline 
\multicolumn{2}{|c|}{Guard class} & \multicolumn{4}{|c|}{Polygon class} \\
\hline
Inclusion & Type & Simple & Ortho & Mono & Star \\
\hline
     & Edge & No & Yes & No & No \\
\cline{2-6}
Open & Diagonal & No & No & No & No \\
\cline{2-6}
     & Mobile & No & Yes & Yes & Yes \\
\hline
     & Edge & No & No & No & No \\
\cline{2-6}
Closed & Diagonal & No & No & No & No \\
\cline{2-6}
     & Mobile & No & No & No & ? \\
\hline
\end{tabular}
\caption{New results in this paper. Entries indicate whether a hidden guard set exists for every polygon in the class.}
\label{tab:results}
\end{table}

\section{Open edge guards}

Recall open edge guards are edges of the polygon excluding the endpoints.

\begin{lemma}
There exists a monotone polygon that does not admit a hidden open edge guard set.
\end{lemma}

\begin{figure}[ht]
\centering
\includegraphics[width=.6\columnwidth]{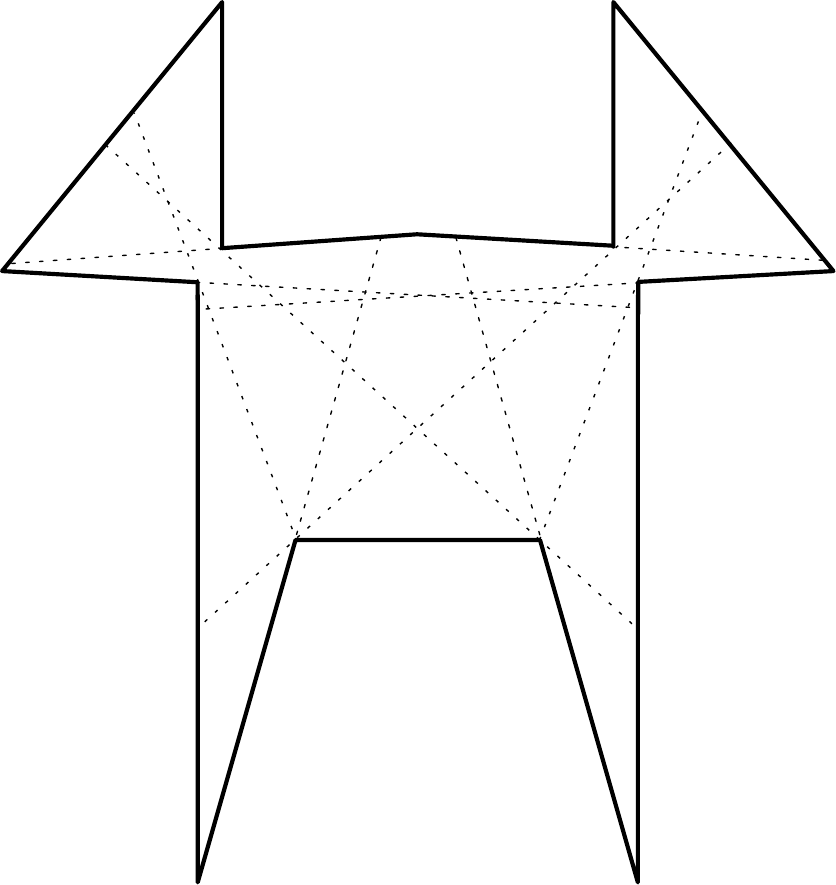}
\caption{A monotone polygon that does not admit a hidden open edge guard set.}
\label{fig:open-edge-monotone-ex}
\end{figure}

\begin{proof}
See Figure~\ref{fig:open-edge-monotone-ex}.
We refer to the convex regions bounded by three edges in the upper left and right portions of the polygon as \emph{ears}. 
Consider guarding the pair of ear regions without using any of the three edges that form each ear.
The cases resulting from these attempts are seen in Figure~\ref{fig:open-edge-monotone-pf-2}.
In each case, any maximal combination of non-ear edges fails to guard either ear completely.
Moreover, a portion of the remaining unguarded region in each ear is not visible from any edge of the other ear.
Thus any guard set contains one of the three edges in each ear.
Also, every pair of ear edges in the same ear see each other, so any guard set contains exactly one edge in each ear.

\begin{figure}[ht]
\centering
\includegraphics[width=1.0\columnwidth]{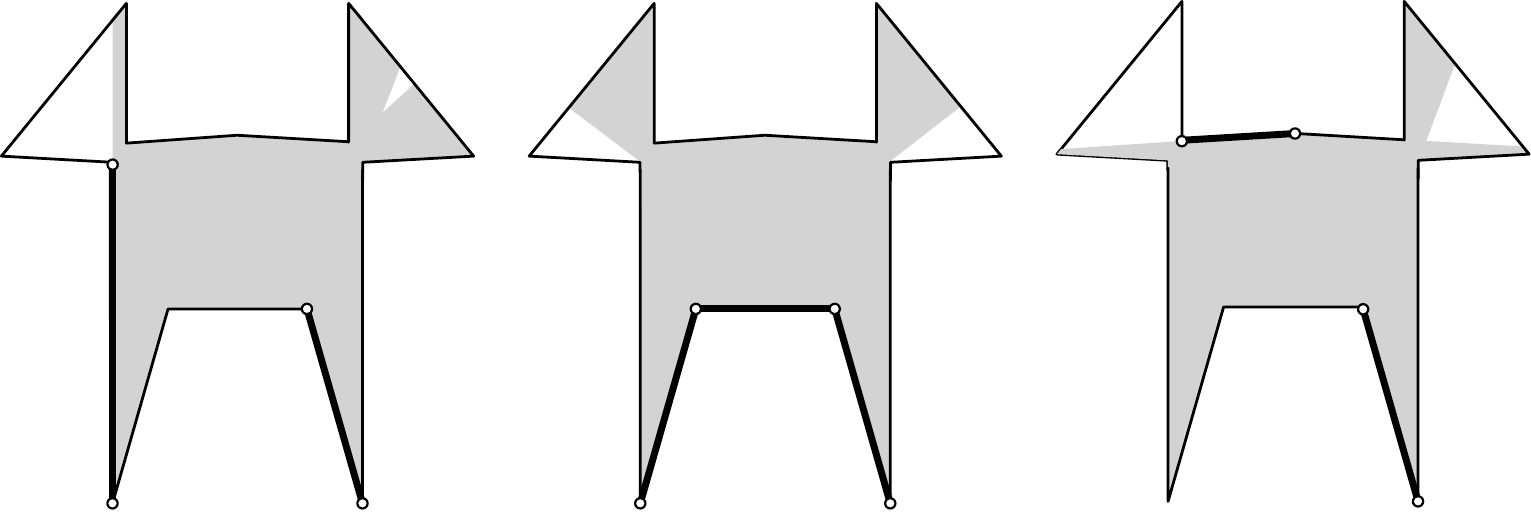}
\caption{All maximal combinations of open edge guards that exclude the six ear edges.} 
\label{fig:open-edge-monotone-pf-2}
\end{figure}

Next, consider possible ear-edge pairs containing one edge from each ear.
In Figure~\ref{fig:open-edge-monotone-pf-1} it is shown that for each such ear-edge pair, the pair cannot be augmented to form a hidden open edge guard set for the polygon.
Thus the polygon cannot be guarded with hidden open edge guards. 

\begin{figure}[ht]
\centering
\includegraphics[width=1.0\columnwidth]{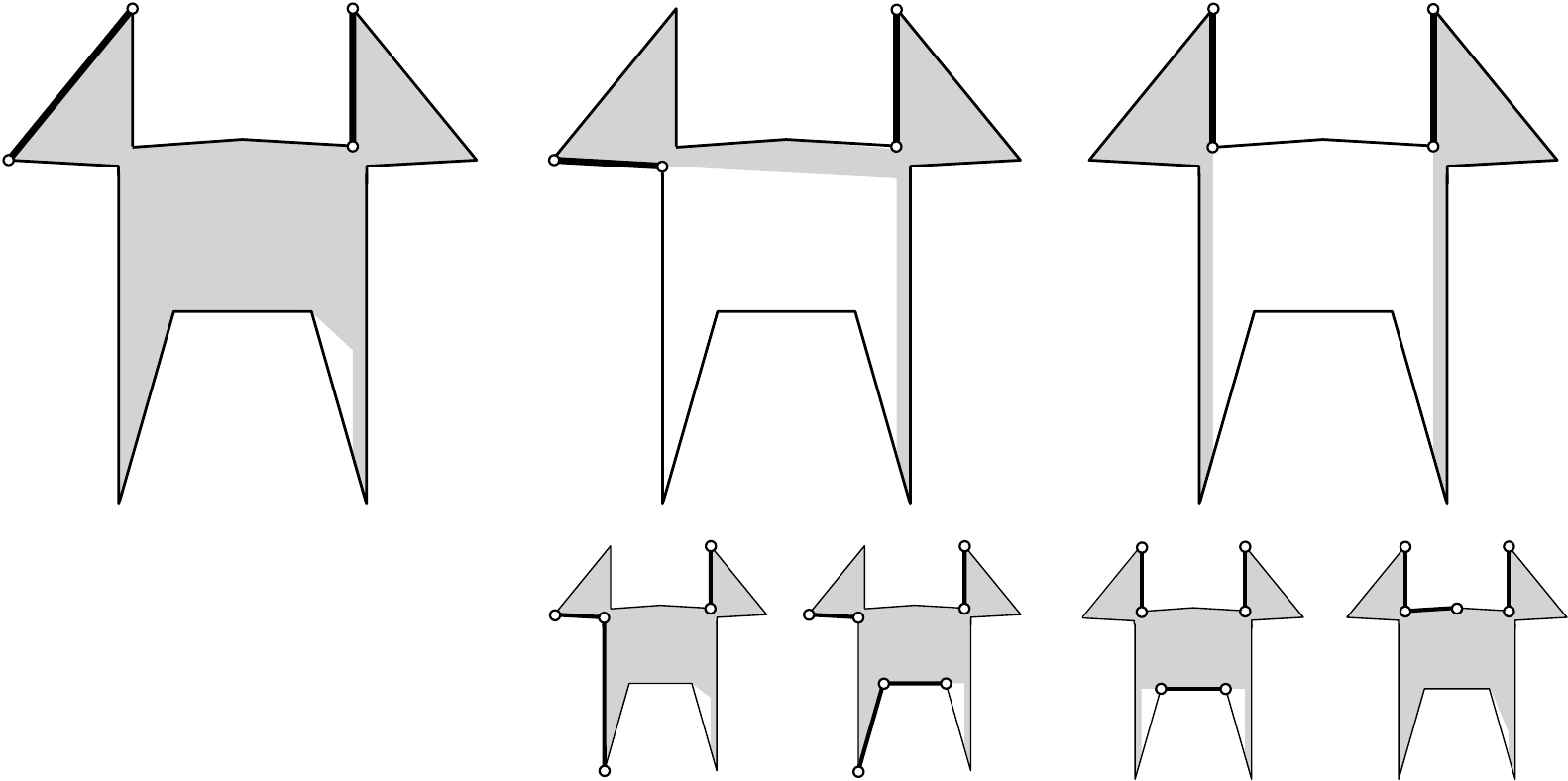}
\caption{All combinations of ear edge pairs and the maximal hidden sets containing each ear edge pair.}
\label{fig:open-edge-monotone-pf-1}
\end{figure}

\end{proof}

\begin{lemma}
There exists a starshaped polygon that does not admit a hidden open edge guard set.
\end{lemma}

\begin{figure}[ht]
\centering
\includegraphics[width=1.0\columnwidth]{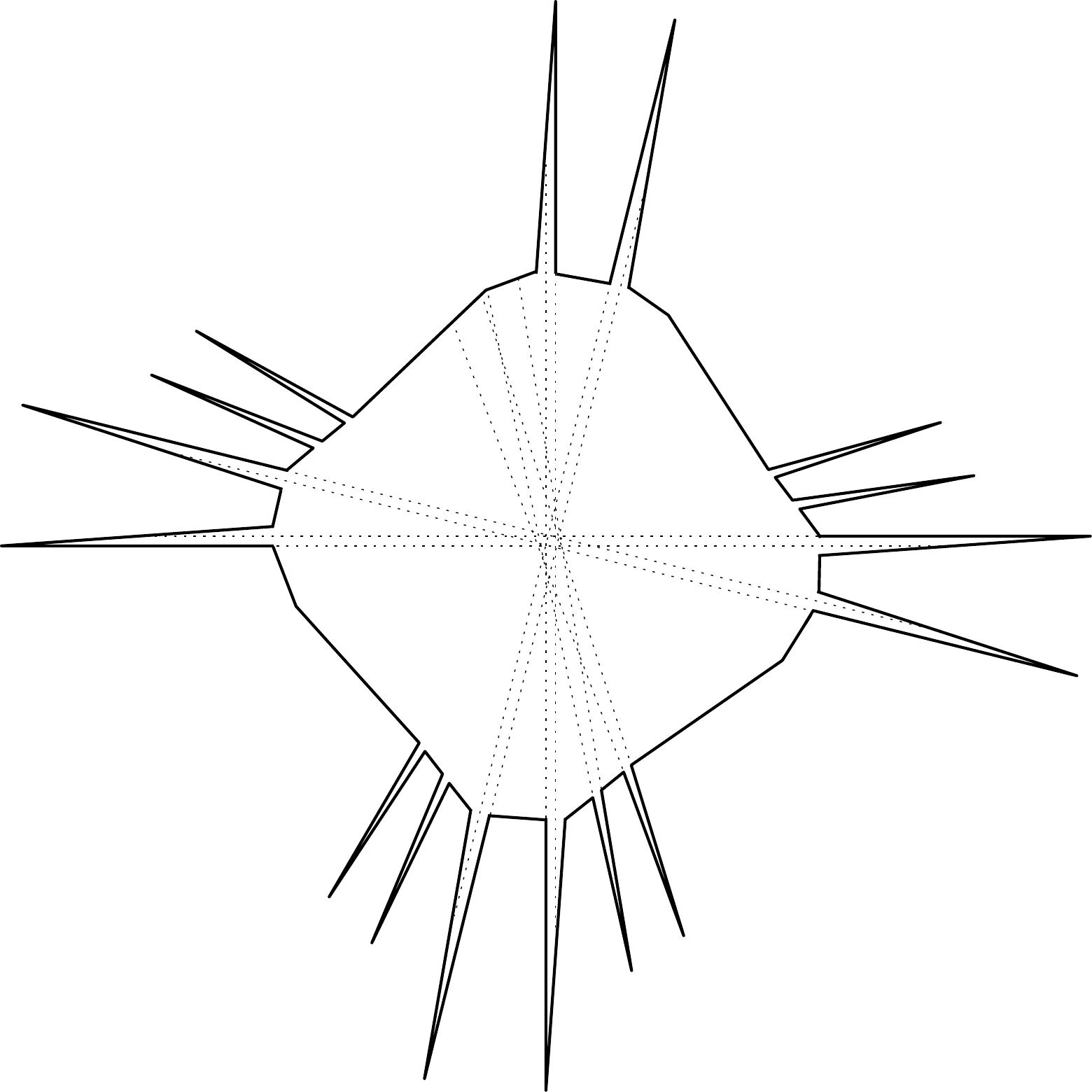}
\caption{A starshaped polygon that does not admit a hidden open edge guard set.}
\label{fig:open-edge-ss-ex}
\end{figure}

\begin{proof}

See Figure~\ref{fig:open-edge-ss-ex}.
The polygon consists of a central convex region with numerous spikes emanating from it.
Figure~\ref{fig:open-edge-ss-pf-2} provides a labeled version of the polygon, with two sets of four large spikes each ($\{a_i\}$ and $\{b_i\}$) and four sets of two small spikes each ($\{c_1, c_2\}$ forms one such set).
Call edges on the central convex region \emph{central edges} and the spike pairs $\{a_1, a_3\}$, $\{a_2, a_4\}$, $\{b_1, b_3\}$, $\{b_2, b_4\}$ \emph{opposing spike pairs}.
Consider guarding the four spikes $\{a_i\}$ without using central edges (see Figure~\ref{fig:open-edge-ss-pf-1}).

\begin{figure}[ht]
\centering
\includegraphics[width=1.0\columnwidth]{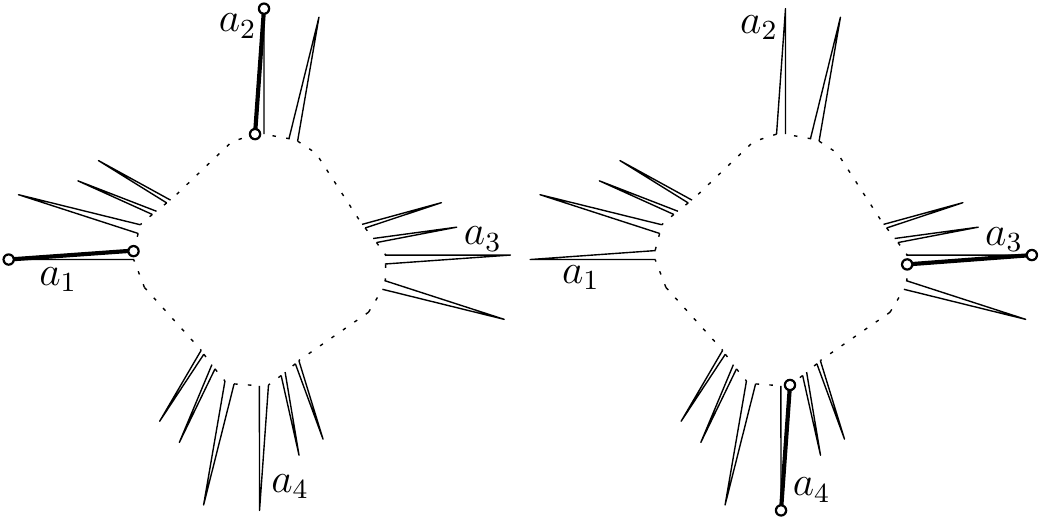}
\caption{The two possible guardings of the four spike $\{a_i\}$ without using central edges (dotted).}
\label{fig:open-edge-ss-pf-1}
\end{figure}

Only one edge per opposing spike pair may be in any hidden edge guard set, as all four edges of an opposing spike pair see each other.
Each spike has two asymmetric edges; one is able to guard the entire opposing spike pair, while the other is not.
Each spike also contains a location not seem by any spike edge not in the spike's opposing spike pair.
Finally, a pair of edges from $a_1$ and $a_4$ see each other, as do a pair in $a_2$ and $a_3$.
So any hidden edge guard set for the opposing spike pairs $\{a_1, a_3\}$ and $\{a_2, a_4\}$ that does not include central edges consists of one of two pairs seen in Figure~\ref{fig:open-edge-ss-pf-1}.

Now consider guarding the entire polygon.
Any central edge guards the interior of at most one spike from $\{a_i\}$ or $\{b_i\}$. 
So one of the two spike sets $\{a_i\}$ and $\{b_i\}$ must be guarded without using central edges.
Without loss of generality, assume the $\{a_i\}$ set is guarded in this way.
Then one of the pairs of edges seen in Figure~\ref{fig:open-edge-ss-pf-1} must be in the guard set.
Again without loss of generality, assume the edge pair of $a_1$ and $a_2$ are selected, as in Figure~\ref{fig:open-edge-ss-pf-2}.
Then there exist two spikes $c_1$ and $c_2$ whose edges are both seen by the guard edges in spikes $a_1$ and $a_2$, but portions of the interiors of $c_1$ and $c_2$ remain unguarded.
The only edges sufficient to guard the interiors of $c_1$ and $c_2$ are the central edges $e_1$ and $e_2$.
However, $e_1$ and $e_2$ each guard the interior of only one spike.
Thus a portion of the interior of either $c_1$ or $c_2$ must remain unguarded, and the polygon cannot be guarded using hidden open edge guards.

\begin{figure}[ht]
\centering
\includegraphics[width=1.0\columnwidth]{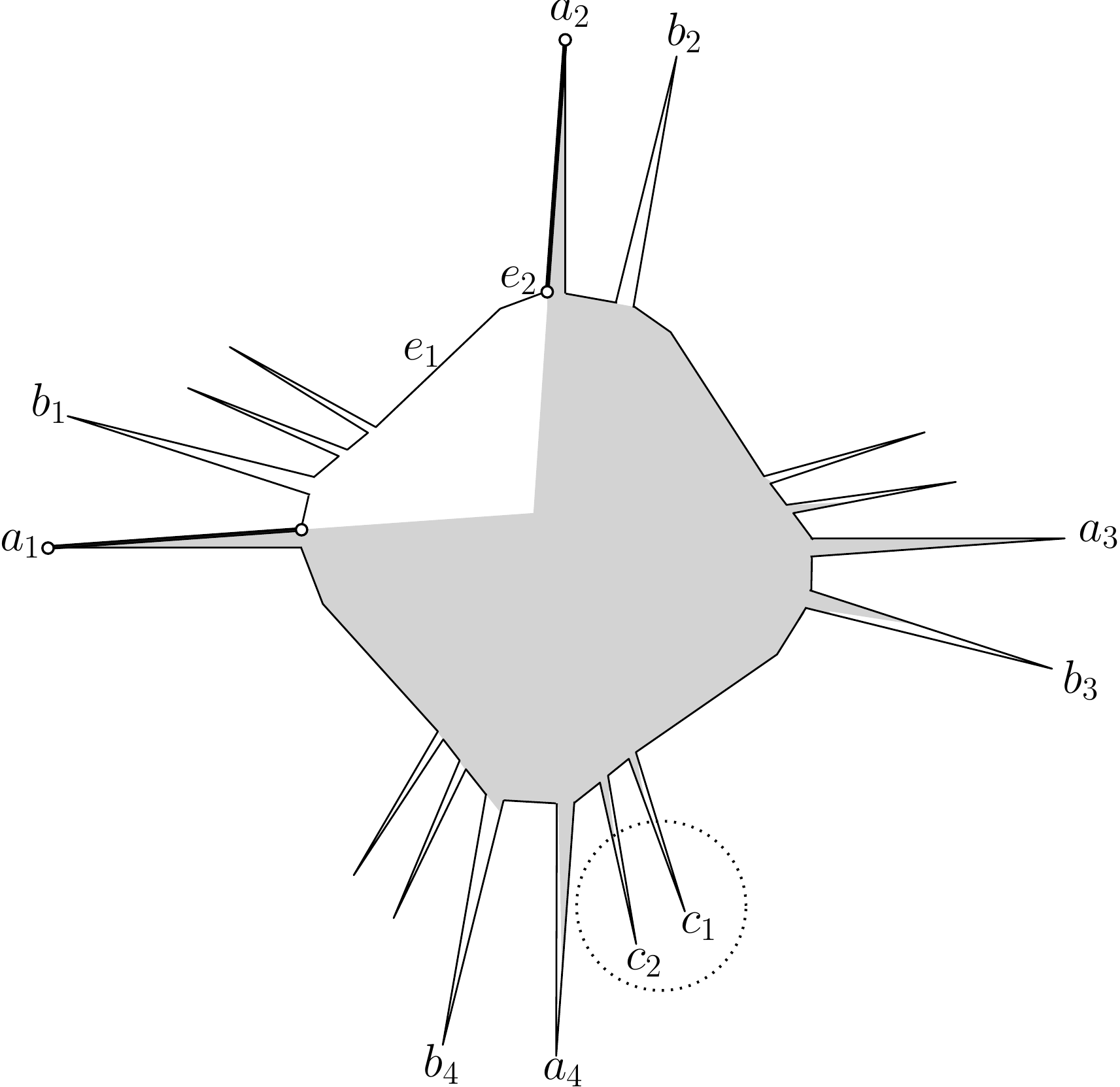}
\caption{A incomplete but necessary set of guard edges and the region they guard. 
The interiors of $c_1$ and $c_2$ remain partially unguarded and cannot be guarded with a hidden open edge set.}
\label{fig:open-edge-ss-pf-2}
\end{figure}

\end{proof}

\begin{lemma}
\label{lem:hidden-open-edge-ortho}
Every orthogonal polygon admits a hidden open edge guard set.
\end{lemma}

\begin{proof}
For an input orthogonal polygon $P$, select a guard set $G$ in the following way: choose all edges of $P$ parallel to the x-axis that bound the interior of the polygon from below. 
Given a location $l \in {\rm int}(P)$, shoot a ray from $l$ at angle $-\pi/2$.
If the ray does not intersect a vertex of $\partial P$, then it intersects an edge in $G$.
If the ray does intersect a vertex (call it $v$), shoot two rays at angles $-\pi/2 - \varepsilon$, $\pi/2 + \varepsilon$, with $\varepsilon > 0$ small enough that one of the edges intersected is incident to $v$.
Such a $\varepsilon$ exists because $\partial P$ is simple.
In either case, the segment from $l$ to $\partial P$ formed by the ray implies that the edge intersected can see $l$, so $G$ is a guard set. 

Now consider any pair of guards $g_1, g_2 \in G$.
If the y-coordinates of $g_1$ and $g_2$ are distinct, then any segment from $g_1$ to $g_2$ must leave the polygon entirely.
If the y-coordinates are identical, the interior of a segment from $g_1$ to $g_2$ must intersect a vertex of the edge containing $g_1$ and so $g_1$ cannot see $g_2$.
So $G$ is also a hidden set.
\end{proof}

For an input polygon with $n$ edges, scanning the edges and selecting those with the properties described can be done by a simple $O(n)$ algorithm.

\section{Open diagonal guards}

\begin{lemma}
There exists a monotone and starshaped polygon that does not admit a hidden open diagonal guard set.
\end{lemma}

\begin{figure}[ht]
\centering
\includegraphics[width=0.4\columnwidth]{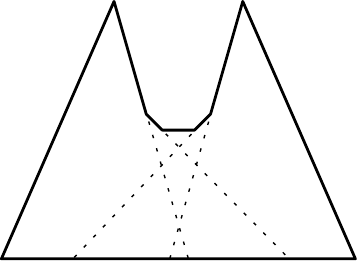}
\caption{A monotone and starshaped polygon that does not admit a hidden open diagonal guard set.}
\label{fig:open-diagonal-ss-ex}
\end{figure}

\begin{proof}
See Figure~\ref{fig:open-diagonal-ss-ex} for an example.
No single diagonal is a guard set for the polygon, as it cannot see both regions near the uppermost pair of vertices.
Every diagonal has an endpoint in common with the lowermost horizontal edge of the polygon, and every diagonal has its other endpoint at a vertex above this edge.
As a result, any pair of diagonals see each other via a horizontal line just above the lowermost horizontal edge of the polygon.
So no set of multiple diagonals can form a hidden set.
\end{proof}

\begin{lemma}
There exists an orthogonal polygon that does not admit a hidden open diagonal guard set.
\end{lemma}

\begin{figure}[ht]
\centering
\includegraphics[width=0.6\columnwidth]{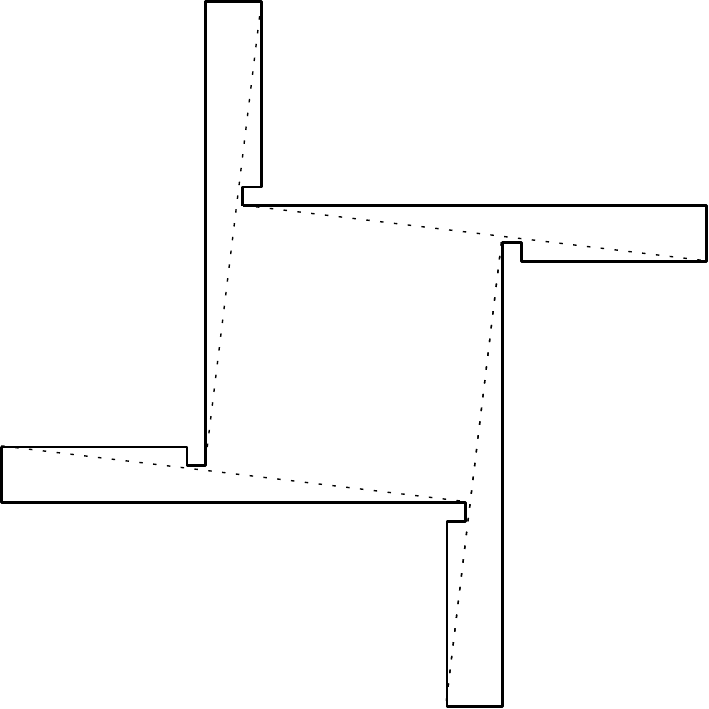}
\caption{An orthogonal polygon that does not admit a hidden open diagonal guard set.}
\label{fig:open-diagonal-ortho-ex}
\end{figure}

\begin{proof}
See Figure~\ref{fig:open-diagonal-ortho-ex}.
Any guard set must see the location in the center of the central convex region (location $l$ in Figure~\ref{fig:open-diagonal-ortho-pf-1}).
This point is only visible from a diagonal that has at least one endpoint on the central convex region.
The four possible diagonals with such an endpoint (up to symmetry) that guard $l$ are shown in Figure~\ref{fig:open-diagonal-ortho-pf-1}.

\begin{figure}[ht]
\centering
\includegraphics[width=1.0\columnwidth]{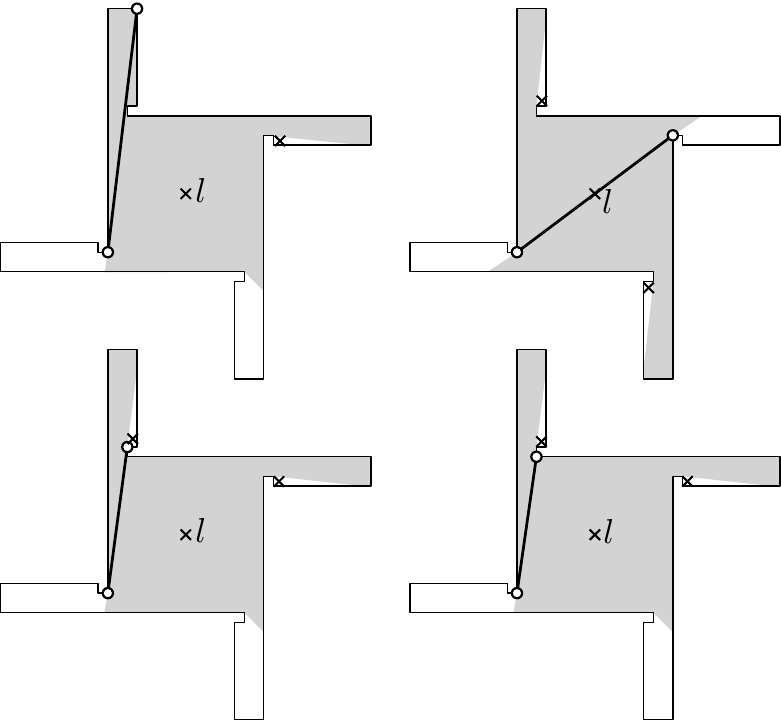}
\caption{The possibilities for selecting a open diagonal guard that sees $l$, and the resulting unguardable regions (``X'' locations in the ears of the polygon).}
\label{fig:open-diagonal-ortho-pf-1}
\end{figure}

In the first case, the diagonal sees the region neighboring $p_2$ and all diagonals that see the region neighboring $p_3$, but fails to guard the entire region around $p_3$.
In the next three cases, $p_1$ sees all diagonals that could be used to guard the region neighboring $p_2$, but does not see the region around $p_2$ entirely.
\end{proof}

\section{Open mobile guards}
\label{sec:open-mobiles}

\begin{lemma}
There exists a simple polygon that does not admit a hidden open mobile guard set.
\end{lemma}

\begin{figure}[ht]
\centering
\includegraphics[width=1.0\columnwidth]{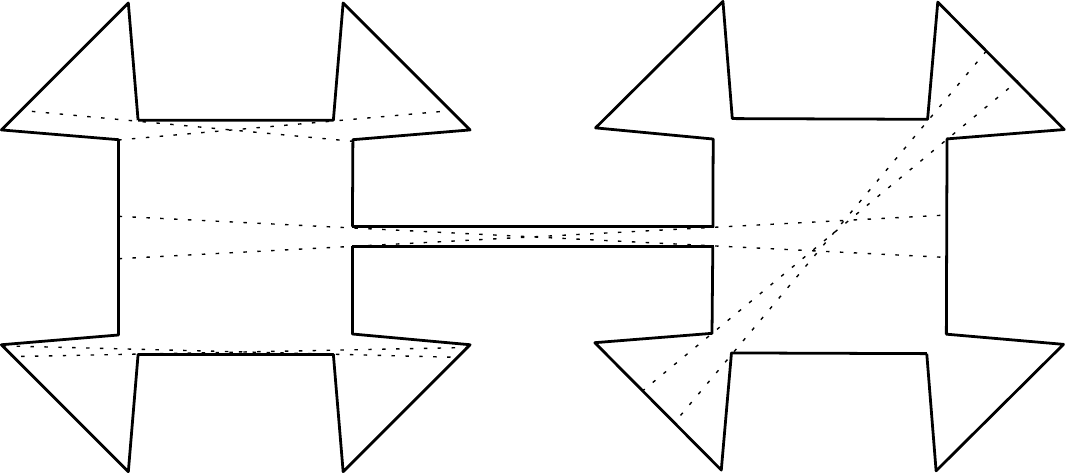}
\caption{A simple polygon that does not admit a hidden open mobile guard set.}
\label{fig:open-mobile-simple-ex}
\end{figure}

\begin{proof}
See Figure~\ref{fig:open-mobile-simple-ex}.
The polygon consists of a narrow \emph{tunnel} connecting left and right clover-shaped \emph{gadgets} that are identical and each have $\pi/2$ rotational symmetry.
Consider guarding the the narrow \emph{tunnel} connecting the left and right clover-shaped regions (which we call \emph{gadgets}).
Guarding the tunnel must occur via either an edge or diagonal in the tunnel or a guard in the left or right gadget.
This results in either the left or right gadget falling into one of the three cases seen in Figure~\ref{fig:open-mobile-simple-pf-1}. 

\begin{figure}[ht]
\centering
\includegraphics[width=1.0\columnwidth]{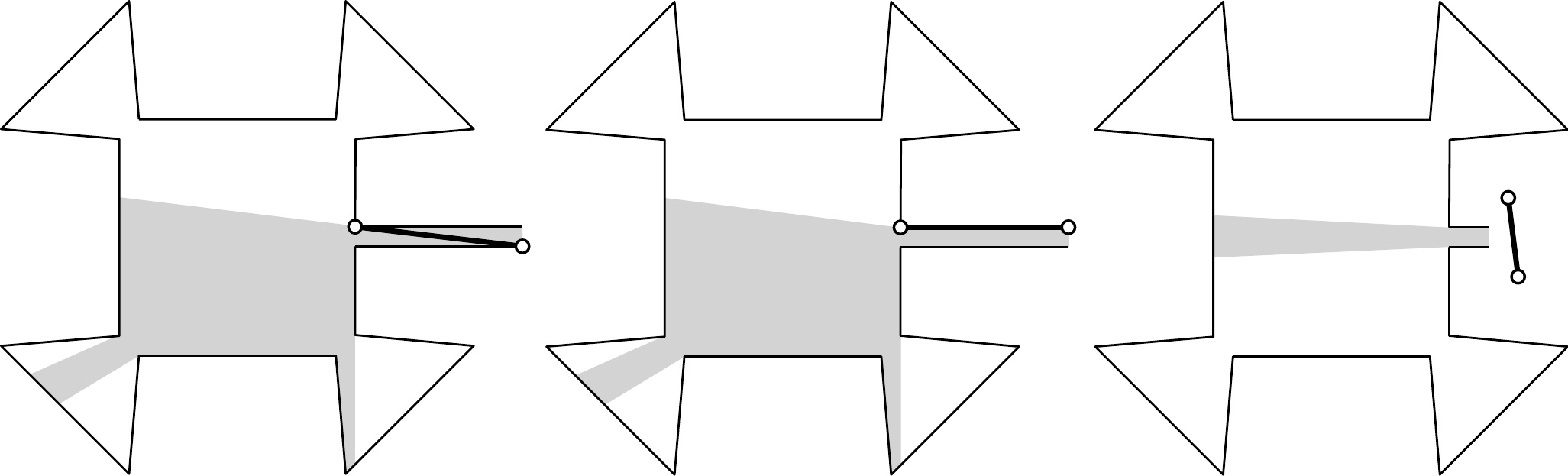}
\caption{The three cases resulting from guarding the tunnel.}
\label{fig:open-mobile-simple-pf-1}
\end{figure}

Because the guard in the left or right gadget can only see locations in the other gadget through the tunnel, the placement of the guard seeing through the tunnel cuts off further illumination of one gadget from guards in the other.
This yields a subproblem of guarding the remainder of either the left or right gadget using only guards in the gadget.

\begin{figure}[ht]
\centering
\includegraphics[width=1.0\columnwidth]{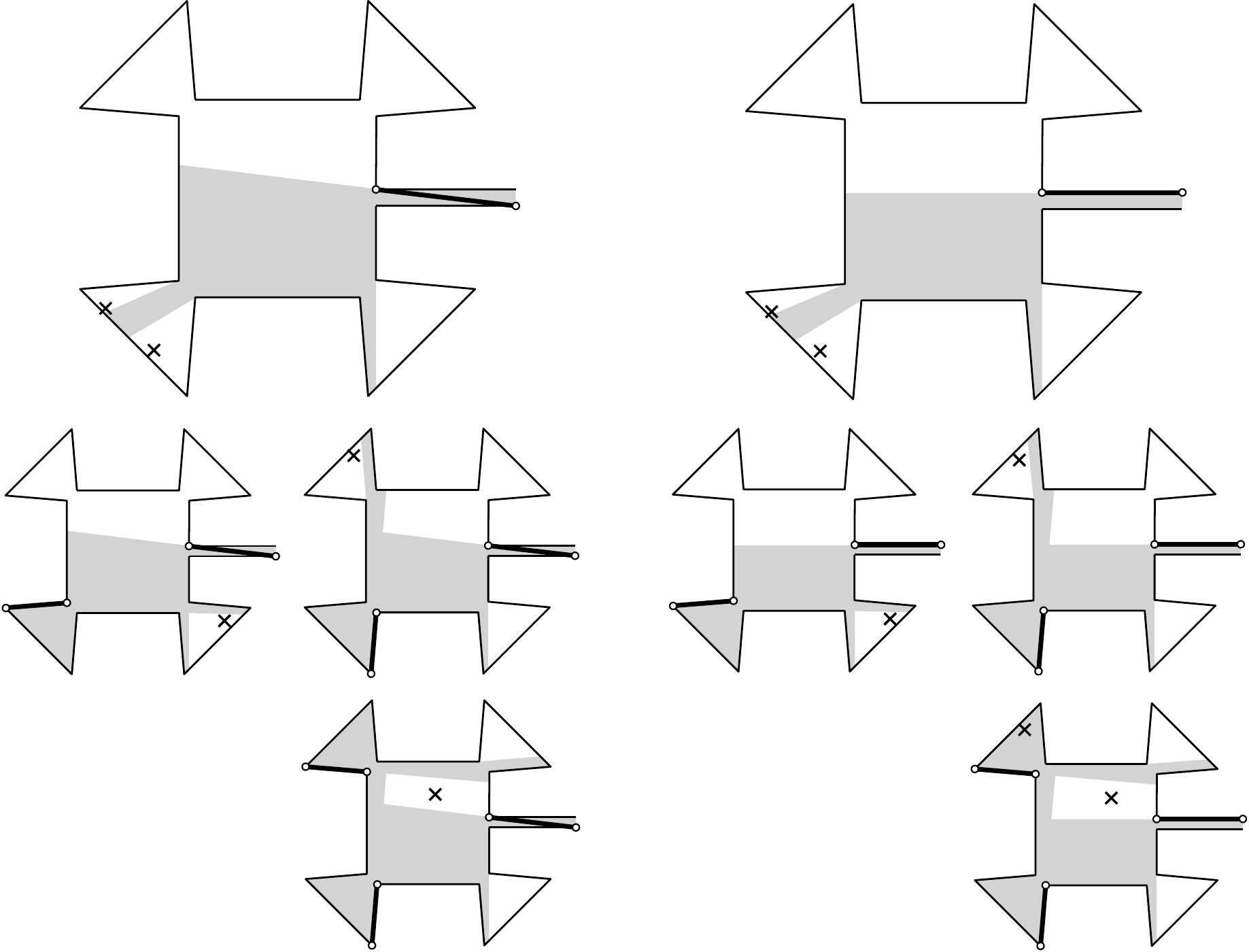}
\caption{Analysis of the first two cases of Figure~\ref{fig:open-mobile-simple-pf-1}.
In both cases, two minimal hidden guard sets see a pair of ``X'' locations in the lower left ear, yielding two subcases.
One subcase results in an unguardable region in the lower right ear.
The other subcase results in an unguardable region in the center of the polygon after selecting an edge guard in the upper left ear needed to see an ``X'' location in the ear.} 
\label{fig:open-mobile-simple-pf-2}
\end{figure}

The first two cases of placing tunnel guards are analyzed in Figure~\ref{fig:open-mobile-simple-pf-2}.
The third case is analyzed in two phases: first, the subcases involving the use of a diagonal guard or an edge guard incident to the central convex region of the gadget are enumerated and each case is shown to yield a set of two or three points that cannot be guarded using any hidden guard set containing the selected guard (Figure~\ref{fig:open-mobile-simple-pf-3}).
Second, the subcases of guarding the gadget using only guards not considered in the first phase (edge guards in the ears of the gadget) are enumerated.
In some of these cases, a sequence of necessary guards are selected in the situation that a location is seen by only one guard hidden from previously-selected guards.
 
\begin{figure}[ht]
\centering
\includegraphics[width=1.0\columnwidth]{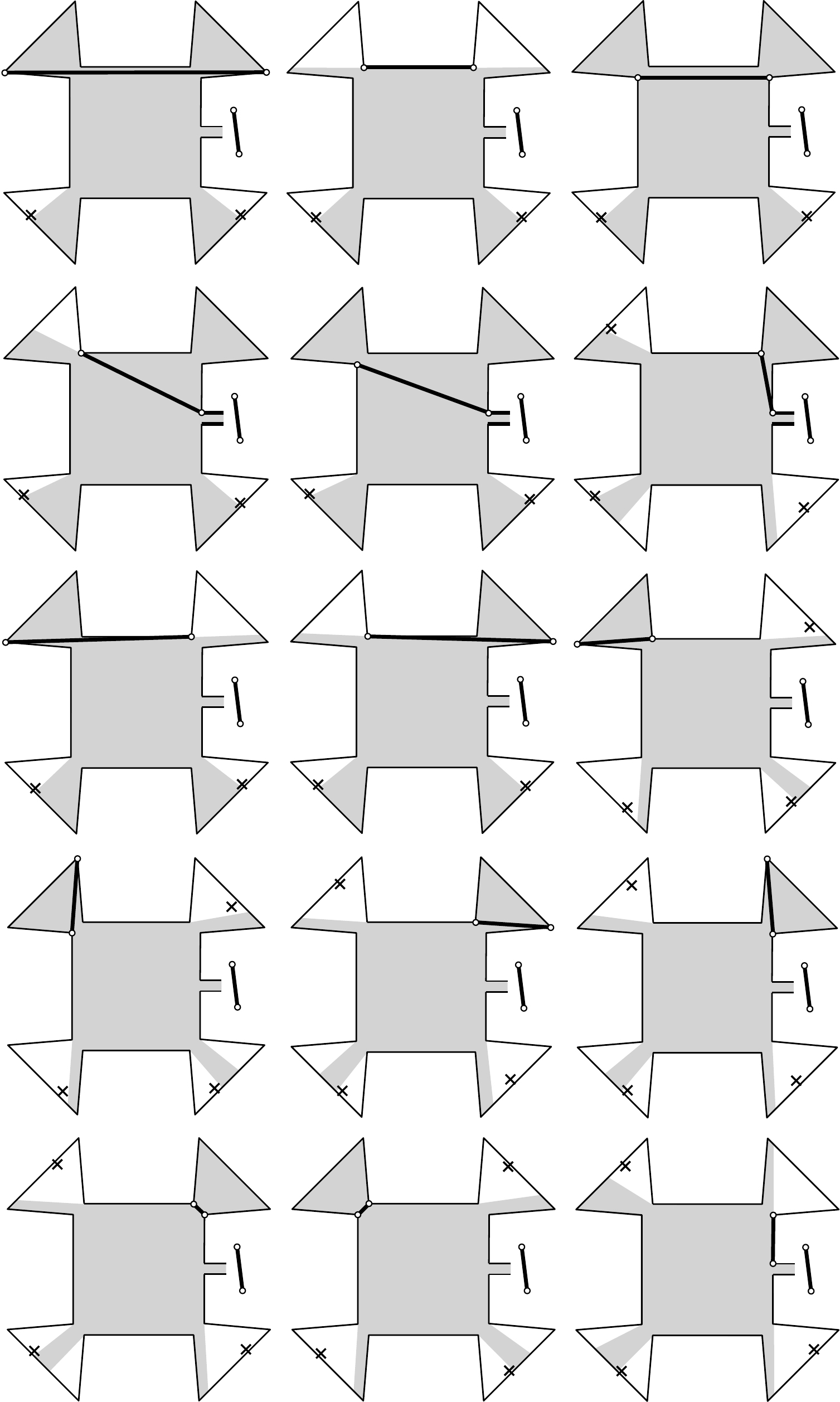}
\caption{The first phase of subcase analysis of the third case of Figure~\ref{fig:open-mobile-simple-pf-1}.
Each case enumerates a possible diagonal of the polygon, or an edge incident to the center convex region.
The ``X'' locations in each scenario form a set of points that cannot all be guarded using a hidden guard set containing the existing guard.}
\label{fig:open-mobile-simple-pf-3}
\end{figure}

\begin{figure}[ht]
\centering
\includegraphics[width=1.0\columnwidth]{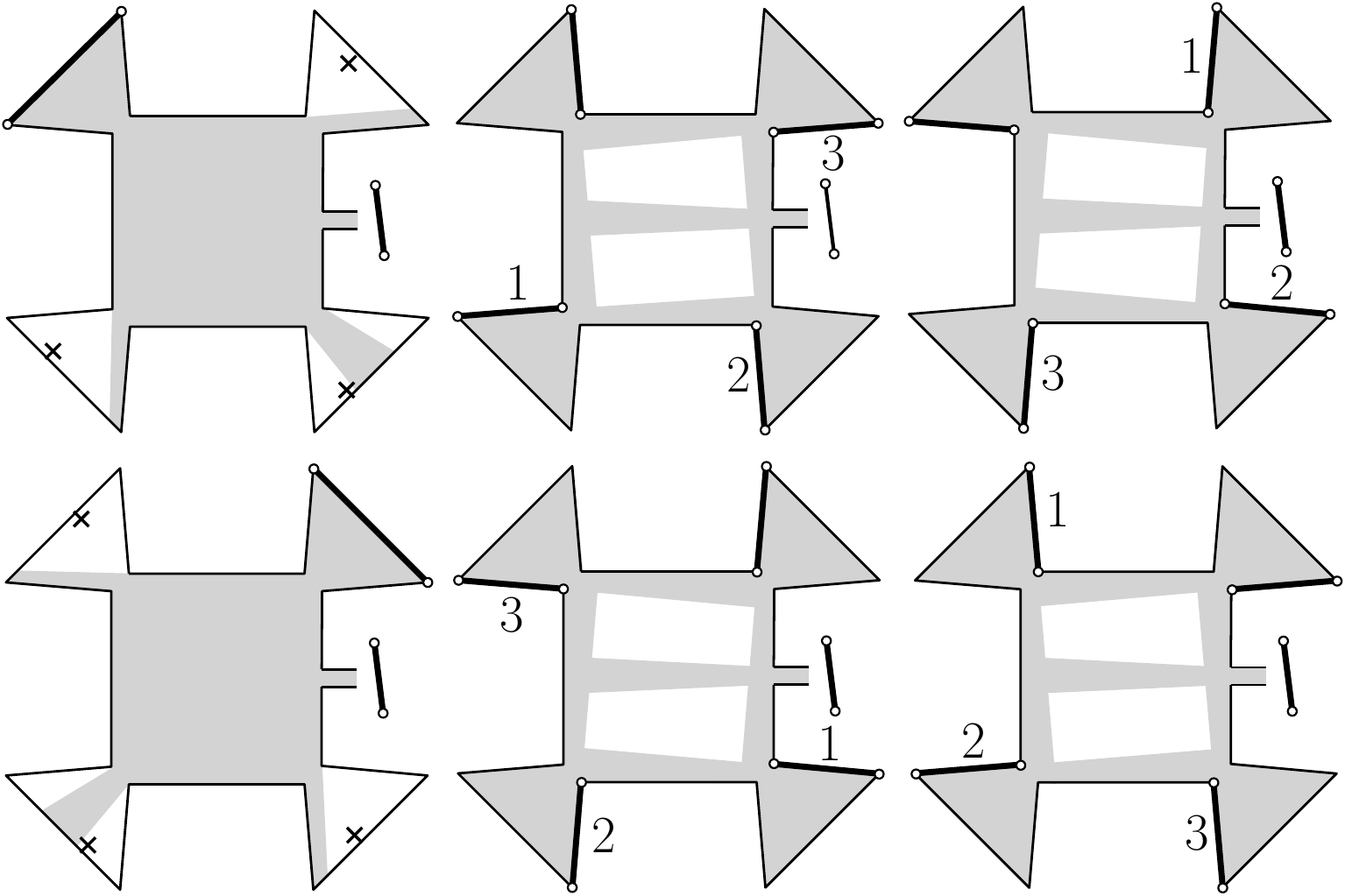}
\caption{The second phase of subcase analysis of the third case of Figure~\ref{fig:open-mobile-simple-pf-1}.
Each case enumerates a possible selection of an ear edge and resulting sequence of necessary edge guards selected.
In the two left cases, the three ``X'' locations form a set that cannot be guarded using any hidden guard set containing the initial guard.
In the remaining cases, necessary guards are selected (numbered in order of selection) until a maximal guard set is formed with regions remaining unguarded.}
\label{fig:open-mobile-simple-pf-4}
\end{figure}

\end{proof}

\begin{obs}
\label{obs:geodesic-hidden}
Let $g$ be a geodesic path between a pair of vertices $p, q$ in a polygon $P$.
Then the interiors of the edges $g$ form a set of hidden open mobile guards in $P$.
\end{obs}

We refer to such a guard set for a path $g$ as the \emph{open mobile guard set induced by $g$}.

\begin{lemma}
\label{lem:open-mobile-monotone}
Every monotone polygon admits a hidden open mobile guard set.
\end{lemma}

\begin{proof}
Assume without loss of generality that the input polygon $P$ is x-monotone, i.e. the intersection of every vertical line with the polygon consists of a single line segment $s$.
Consider a geodesic $g$ path between a pair of left- and rightmost vertices in $P$.
Since $g$ is a path between extreme vertices in the x-direction, $s$ contains a point $p \in g$.
As $p$ sees all of $s$, $g$ guards the entirety of $P$.

Let $G$ be the hidden open mobile guard set induced by $g$.
The vertical segment $s$ intersects either a point on a guard in $G$, or a vertex in $V$.
If $s$ intersects a vertex $v \in V$, then since the boundary of $P$ is simple (i.e. non-intersecting) there must be a nearly-vertical segment for each location $l \in s$ connecting $l$ to a location in the interior of an edge of $g$ without intersecting $\partial P$.
So $G$ forms a guard set for $P$.
\end{proof}

For a polygon with $n$ edges, computing such a guard set can be done in $O(n)$ time: finding left- and rightmost points in $P$ takes $O(n)$ time, and finding the shortest path between them also takes $O(n)$ by combining the results of Garey et al.~\cite{Garey-1978} and Guibas et al.~\cite{Guibas-1987}.

A natural approach to finding an open mobile guard set for a starshaped polygon is to look for a mobile guard that intersects the \emph{kernel} of the polygon.
Unfortunately such a guard may not exist as noted in~\cite{ORourke-1987} (see Figure~\ref{fig:no-mobile-in-kernel}). 

\begin{figure}[ht]
\centering
\includegraphics[width=0.4\columnwidth]{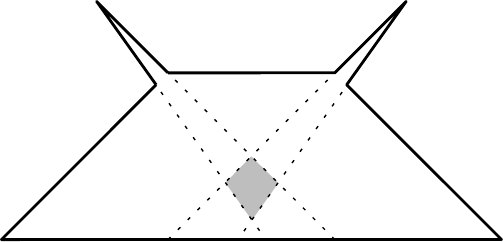}
\caption{A starshaped polygon with no edge or diagonal intersecting its kernel (gray).}
\label{fig:no-mobile-in-kernel}
\end{figure}

The following lemma is used in the proof of Lemma~\ref{lem:open-mobile-ss}. 

\begin{lemma}
\label{lem:a-mobile-sees-all}
Let $P$ be a starshaped polygon translated so that the origin lies in the kernel of $P$, and let $v, v'$ be consecutive reflex vertices such that angle between the rays from $v$ and $v'$ through the origin (sweeping from $v$ to $v'$) is at most $\pi$.
If a geodesic path $g \in P$ intersects both rays either before or after they intersect the origin, then $g$ guards the subpolygon $R$ bounded by the portions of the two rays before they intersect the origin, and the portion of $\partial P$ from $v$ to $v'$.
\end{lemma}

\begin{proof}
See Figure~\ref{fig:a-mobile-sees-all}.
Because $k$ sees every point in $P$, the rays form angles of at most $\pi$ with both edges incident to $v$ and $v'$ on $\partial P$, and do not intersect $\partial P$ before reaching the origin.
Define the points on $\partial P$ intersected by the rays from $v$ and $v'$ as $i$ and $i'$, respectively.
Let $l$ be a location in the subpolygon $R$.

\begin{figure}[ht]
\centering
\includegraphics[width=0.6\columnwidth]{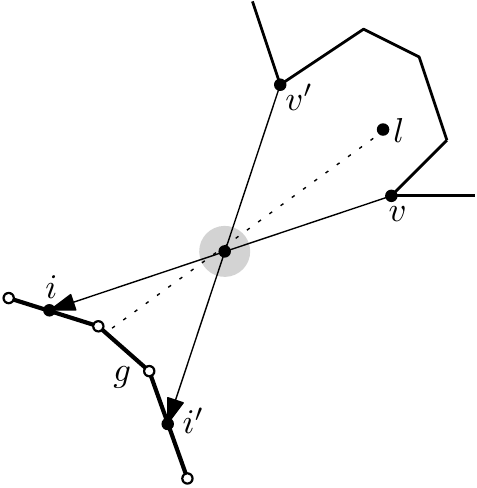}
\caption{A schematic for the proof of Lemma~\ref{lem:a-mobile-sees-all}.}
\label{fig:a-mobile-sees-all}
\end{figure}

Recall that the portion of $\partial P$ between $v$ and $v'$ has only convex vertices, and the angle formed by the two rays at the origin is at most $\pi$ by assumption, so $R$ is convex.
If $g$ intersects the rays before they intersect the origin, then the interior of an edge in $g$ must intersect $R$ and $l$ is seen.
Now suppose that $g$ intersects the rays after they leave the origin.
Because $g \in P$ and the boundary of the polygon is simple, there exists some $\varepsilon > 0$ such that any ray from $l$ through an $\varepsilon$-disk around the origin intersects $g$ before (or possibly as) the ray intersects $\partial P$.
So a point on a guard in the hidden open mobile guard set induced by $g$ sees $l$ and thus this set is a guard set for $R$.
\end{proof}

\begin{lemma}
\label{lem:open-mobile-ss}
Every starshaped polygon admits a hidden open mobile guard set.
\end{lemma}

\begin{proof}
Let $P$ be a given starshaped polygon translated so that the origin lies in the kernel of $P$.
Consider shooting rays from each reflex vertex through the origin as seen in the left portion of Figure~\ref{fig:ray-wheel-1}.
Find a double wedge $W$ formed by a consecutive pair of these rays such that each wedge is coincident to exactly one reflex vertex (which we call $u$ and $u'$) as seen in right portion of Figure~\ref{fig:ray-wheel-1} as a dark gray region. 
Such a double wedge is formed by every pair of consecutive intersections of rays along $\partial P$ such that one intersection is the start of a ray (at a reflex vertex of $P$), and the other is the termination of a ray.

\begin{figure}[ht]
\centering
\includegraphics[width=1.0\columnwidth]{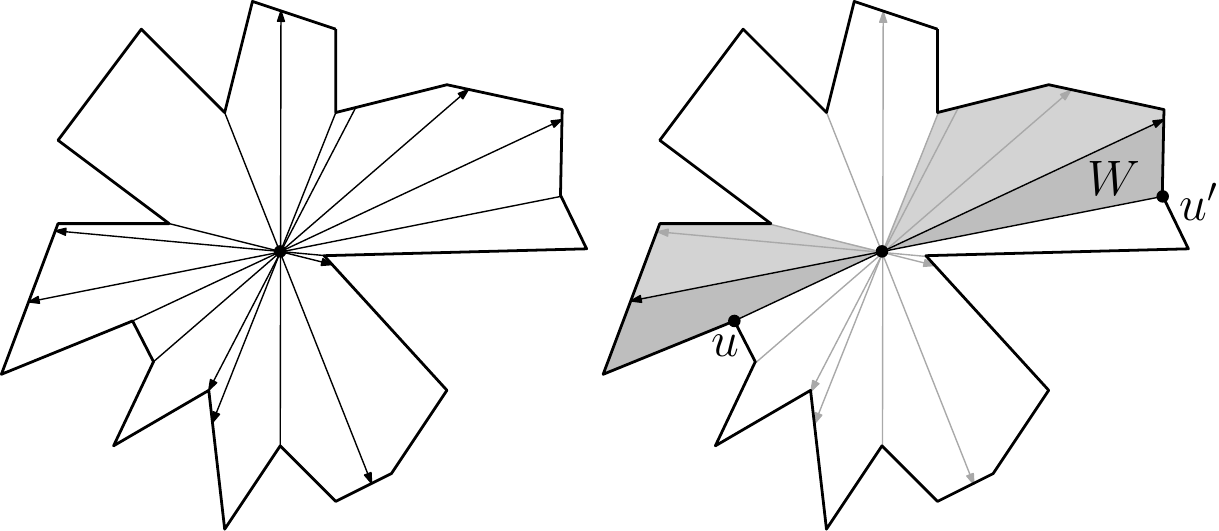}
\caption{Left: a starshaped polygon with rays from each reflex vertex through the origin.
Right: the polygon and a double wedge $W$ (dark gray) with one reflex vertex ($u$ or $u'$) incident to each wedge.
The light and dark gray regions together form the subpolygons possibly left unguarded by the hidden open mobile guard set induced by a geodesic path from $u$ to $u'$.}
\label{fig:ray-wheel-1}
\end{figure}

For every consecutive pair of reflex vertices $v, v'$ on $\partial P$, the rays from $v$ and $v'$ through the origin lie entirely in $P - W$.
Two pairs are an exception: the two pairs containing $u$ and $u'$ that form a pair of wedges, each containing half of the double wedge $W$ (seen as the dark gray double wedge extended with two light gray wedges in the right portion of Figure~\ref{fig:ray-wheel-1}). 
For all remaining pairs, the geodesic path from $u$ to $u'$ intersects both rays either before or after they have passed through the origin.
Therefore, by Lemma~\ref{lem:a-mobile-sees-all}, the hidden open mobile guard set induced by $g$ sees the entire polygon except (possibly) the pair of wedges bounded by two pairs of consecutive reflex vertices adjacent to $u$ and $u'$. 

It may be the case that the two remaining wedges are actually a single non-convex subpolygon with reflex vertex at the origin (see Figure~\ref{fig:non-convex-remainder}).

\begin{figure}[ht]
\centering
\includegraphics[width=0.51\columnwidth]{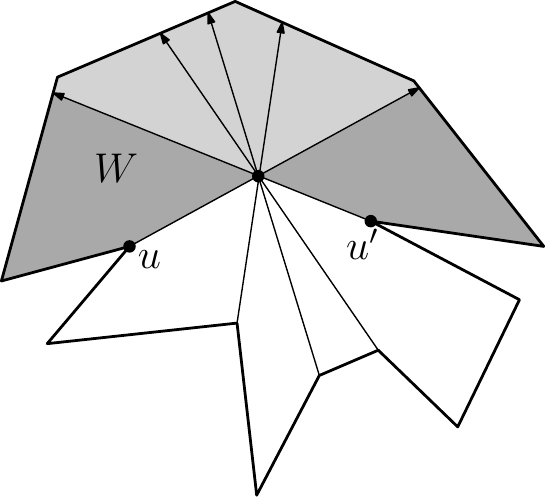}
\caption{A polygon and double wedge $W$ (dark gray region) where the region not necessarily guarded by the hidden open mobile guard set induced by the geodesic path from $u$ to $u'$ is actually a single non-convex polygon (light and dark gray regions combined) bounded by $u$ and $u'$.}
\label{fig:non-convex-remainder}
\end{figure}

In this situation the subpolygon can be bisected into two convex subpolygons by a ray bisecting the reflex angle at the origin.

\begin{figure}[ht]
\centering
\includegraphics[width=1.0\columnwidth]{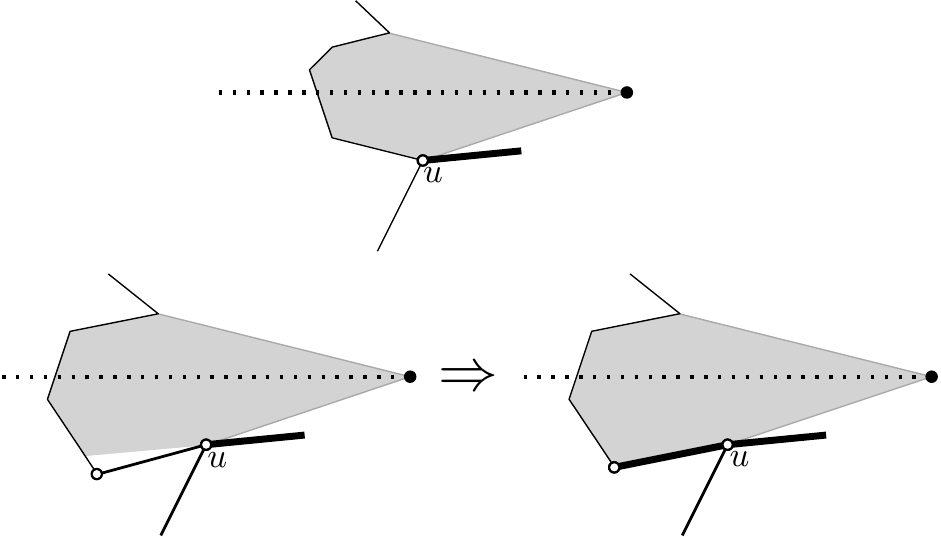}
\caption{The two cases of guarding the remaining subpolygons.
In the case shown in the upper part of the figure, the existing geodesic is sufficient to guard the wedge.
In the second case, the geodesic leaves a portion of the wedge unguarded and must be extended.}
\label{fig:geodesic-corner-case}
\end{figure}

Recall that each convex subpolygon has a vertex $u$ or $u'$ in common with the geodesic's final edge (see Figure~\ref{fig:geodesic-corner-case}).
If the interior angle formed by these two edges is at most $\pi$, then the subpolygon is seen by the interior of the final edge of the geodesic.
If not, the geodesic can be extended to include an edge of $\partial P$ in the subpolygon that guards the subpolygon completely.

\begin{figure}[ht]
\centering
\includegraphics[width=0.51\columnwidth]{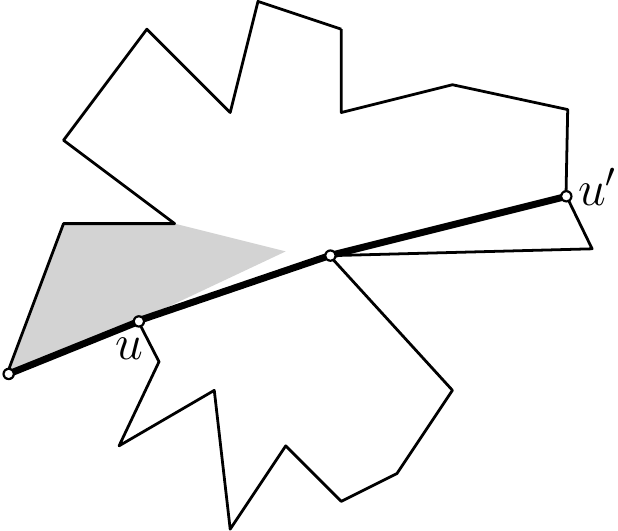}
\caption{A polygon with a geodesic path inducing a hidden open mobile guard set for the polygon.
The initial geodesic from $u$ to $u'$ leaves the gray region incident to $u$ partially unguarded, so the geodesic is extended by one edge.}
\label{fig:ray-wheel-2}
\end{figure}

Thus the hidden mobile guard set induced by the geodesic described guards $P$.
\end{proof}

Computing such a guard set for a polygon with $n$ edges can be done in $O(n)$ time, as each step takes at most $O(n)$ time:
1.~compute a point in the kernel of the polygon ($O(n)$~time by Lee and Preparata~\cite{Lee-1979}).
2.~find a separating angle $\theta$ ($O(n)$ time).
3.~triangulate the polygon and find a geodesic between the reflex vertices $u$ and $u'$ ($O(n)$~time by Fournier and Montuno~\cite{Fournier-1984} and Guibas et al.~\cite{Guibas-1987}).
4.~check whether the two remaining subpolygons are already covered by the geodesic, and extend the geodesic by an additional edge if necessary~($O(1)$~time).

\section{Closed edge and diagonal guards}

In the next section we present orthogonal and monotone polygons that do not admit hidden closed mobile guard sets.
Note that these polygons also serve as examples of polygons that do not admit hidden closed edge or hidden closed diagonal guards.
For starshaped polygons no such example is known.

\begin{lemma}
There exists a starshaped polygon that does not admit a hidden closed edge guard set.
\end{lemma}

\begin{figure}[ht]
\centering
\includegraphics[width=.4\columnwidth]{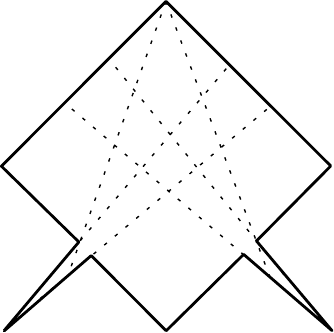}
\caption{A starshaped polygon that does not admit a hidden closed edge guard set.}
\label{fig:closed-edge-ss-ex}
\end{figure}

\begin{proof}
See Figure~\ref{fig:closed-edge-ss-ex}.
Every edge has at least one endpoint on the central convex region, so any hidden edge set has at most one edge.
However, no single edge is sufficient to guard the polygon as no edge guards the interior of both spikes completely.
\end{proof}

\begin{lemma}
There exists a starshaped polygon polygon that does not admit a hidden closed diagonal guard set.
\end{lemma}

\begin{proof}
The polygon in Figure~\ref{fig:open-diagonal-ss-ex} does not admit hidden open diagonal guard set, and also fails to admit a hidden closed diagonal guard set for the same reason: no single diagonal is sufficient to guard the entire polygon and any set of two or more diagonals is not a hidden set. 
\end{proof}

\section{Closed mobile guards}

\begin{lemma}
There exists an orthgonal polygon that does not admit a hidden closed mobile guard set.
\end{lemma}

\begin{figure}[ht]
\centering
\includegraphics[width=.8\columnwidth]{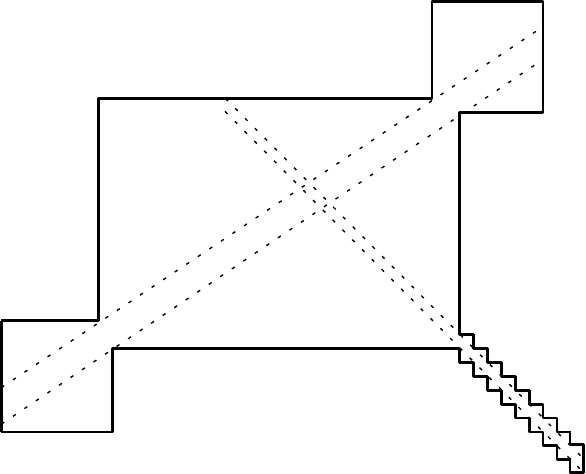}
\caption{An orthogonal polygon that cannot be guarded using hidden closed mobile guards.}
\label{fig:closed-mobile-ortho-ex}
\end{figure}

\begin{proof}
See Figure~\ref{fig:closed-mobile-ortho-ex}.
We refer to the convex regions in the lower left and upper right of the polygon as \emph{ears}, and the narrow region in the lower right of the polygon as the \emph{cave}.
First, consider guarding the cave region.
Any guard that sees a portion consisting of more than four convex vertices in the cave sees a narrow strip of space extending to the end of the cave and containing all reflex vertices in the cave.
As a result, a second guard inside the cave is not permitted, as either: 1. an endpoint of the second guard is seen by the first guard, or 2. the second guard intersects the narrow strip seen by the first guard.
Therefore, a single guard extending the length of the cave is needed, and because no vertex in the remainder of the polygon can be connected to such a guard, this guard ends at a vertex at the mouth of the cave.

\begin{figure}[ht]
\centering
\includegraphics[width=.8\columnwidth]{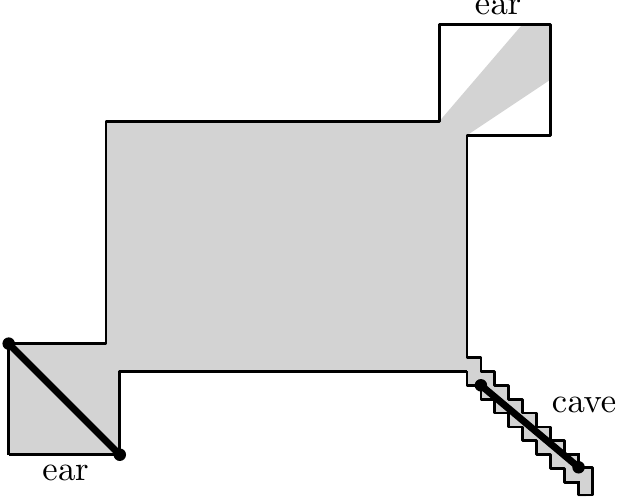}
\caption{Attempting to guard an orthogonal polygon with closed mobile guards. An initial guard spanning the length of the cave is necessary, and a second guard for an ear results in an unguardable region in the other ear.}
\label{fig:closed-mobile-ortho-pf-1}
\end{figure}

Now consider guarding the remainder of the polygon (see Figure~\ref{fig:closed-mobile-ortho-pf-1}).
The guard inside the cave is unable to see either ear.
Guarding both ears using a second guard is forbidden, as such a guard must intersect the region seen by the guard in the cave.
Guarding one ear using a second guard results in an unguarded region in the other ear, and no hidden guard choices remaining.
So the polygon does not admit a hidden closed mobile guard set.
\end{proof}

\begin{lemma}
There exists a monotone polygon that does not admit a hidden closed mobile guard set.
\end{lemma}

\begin{figure}[ht]
\centering
\includegraphics[width=0.9\columnwidth]{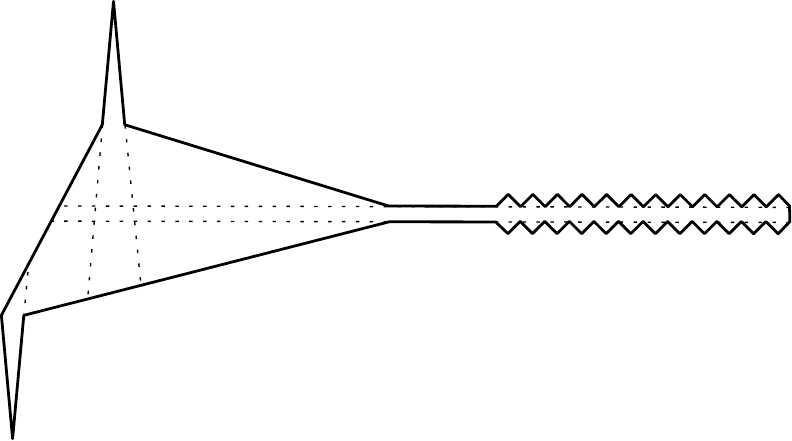}
\caption{A monotone polygon that cannot be guarded using hidden closed mobile guards.}
\label{fig:closed-mobile-monotone-ex}
\end{figure}

\begin{proof}
See Figure~\ref{fig:closed-mobile-monotone-ex}.
The construction is similar to that of Figure~\ref{fig:closed-mobile-ortho-ex}, with a long cave (right portion of the polygon) and pair of ears (lower and upper left portion of the polygon).
Guarding the cave requires using a single diagonal guard that extends the length of the cave (see Figure~\ref{fig:closed-mobile-monotone-pf-1}).
As a result, two guards are needed to guard the triangular pair of ears in the left portion of the polygon.
Any two such guards must see each other, so the polygon does not admit a hidden closed mobile guard set.

\begin{figure}[ht]
\centering
\includegraphics[width=0.9\columnwidth]{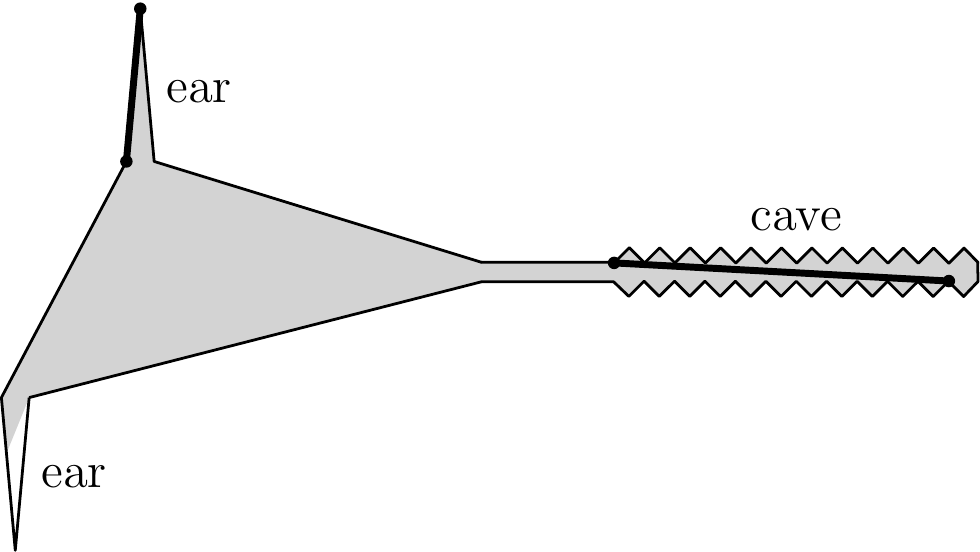}
\caption{Attempting to guard a monotone polygon with hidden closed mobile guards.
An initial guard spanning the length of the cave is necessary, and a second guard for an ear results in an unguardable region in the other ear.}
\label{fig:closed-mobile-monotone-pf-1}
\end{figure}

\end{proof}

We end with the open problem from Table~\ref{tab:results}.

\begin{conj}
Every starshaped polygon admits a hidden closed mobile guard set.
\end{conj}

\section*{Acknowledgements}

We thank Csaba T\'{o}th for helpful discussions and Richard Pollack, Joseph Malkevitch, John Iacono, and Bill Hall for suggesting interesting problems in this area.

\end{document}